\documentclass[12pt]{article}

\usepackage[english]{babel}
\usepackage{amssymb,amsmath,amsthm,graphicx,mathtools}
\usepackage{fullpage}
\usepackage[round]{natbib}

\newtheorem{theorem}{Theorem}
\newtheorem{corollary}{Corollary}

\newtheorem{lemma}{Lemma}
\newtheorem{proposition}{Proposition}

\theoremstyle{remark}
\newtheorem{remark}{Remark}

\theoremstyle{remark}
\newtheorem{example}{Example}

\theoremstyle{remark}

\usepackage{geometry} % to change the page dimensions
\usepackage{graphicx} % support the \in Cludegraphics command and options

\usepackage{booktabs} % for much better looking tables
\usepackage{array} % for better arrays (eg matrices) in maths
\usepackage{paralist} % very flexible & customisable lists (eg. enumerate/itemize, etc.)
\usepackage{verbatim} % adds environment for commenting out blocks of text & for better verbatim
\usepackage{subfig} % make it possible to in Clude more than one captioned figure/table in a single float
\usepackage{fancyhdr} % This should be set AFTER setting up the page geometry
\usepackage{sectsty}
\allsectionsfont{\sffamily\mdseries\upshape} % (See the fntguide.pdf for font help)
\usepackage[nottoc,notlof,notlot]{tocbibind} % Put the bibliography in the ToC
\usepackage[titles,subfigure]{tocloft} % Alter the style of the Table of Contents

\title{Equivalent Choice Functions and Stable Mechanisms}
\author{Jan Christoph Schlegel\footnote{The current paper extends and supersedes a paper entitled "Virtual Demand and Stable Mechanisms", an extended abstract of which appeared in the Proceedings of the 2016 ACM Conference on Economics and Computation (EC’16). Chapter 4 of the original paper has been published as an independent working paper with the title ``Group-Strategy-Proof Mechanisms for Job Matching with Continuous Transfers". I am grateful to my adviser Bettina Klaus for many useful comments that greatly improved the paper and to Sangram Kadam for many insightful discussions. I thank Battal Do\u{g}an, Federico Echenique, Ravi Jagadeesan, Flip Klijn, Fuhito Kojima, Maciej Kotowski, Jordi Mass\'{o}, an anonymous referee, participants of the 2016 Meeting of the Social Choice and Welfare Society, the 17th ACM Conference on Economics and Computation, the 2018 ASSA meetings, the 5th Match-UP workshop, the Lisbon game theory meetings, seminar participants at Bocconi University, in Marseille, in Maastricht, at City University and at Hitotsubashi University for comments on a previous version of the paper. I gratefully acknowledge financial support by the Swiss National Science Foundation (SNSF) under project 100018-150086.
	}\\ \small Department of Economics, City, University of London, United Kingdom\\{\small\tt jansc@alumni.ethz.ch}}
\date{} % Activate to display a given date or no date (if empty),
\begin{document}
	\maketitle
	\begin{abstract}
		We study conditions for the existence of stable and
		(group)-strategy-proof mechanisms in a many-to-one matching model with
		contracts if students' preferences are monotone in contract terms.
		We show that ``equivalence", properly defined, to a choice profile
		under which contracts are substitutes and the law of aggregate demand holds
		is a necessary and sufficient condition for the existence of a
		stable and (group)-strategy-proof mechanism.
		
		Our result can be interpreted as a (weak) embedding result for
		choice functions under which contracts are observable substitutes
		and the observable law of aggregate demand holds.
		\emph{JEL-classification:} C78, D47
	\end{abstract}
	\section{Introduction}
	Centralized clearing houses based on the deferred-acceptance
	mechanism are at the heart of many successful matching
	markets~\cite[]{Roth1984a,AbdulkadirogluSonmez2003,SonmezSwitzer2013,Sonmez2013}.
	Deferred-acceptance mechanisms are appealing because they produce
	stable outcomes, meaning that no subgroup of agents can find a
	mutually beneficial deviation and, thus, would have a reason to
	contract outside the market.\footnote{\cite{Roth.1991} provides
		evidence that clearing houses using unstable mechanisms tend to fail
		in practice.} Moreover, it is safe for the applying side of the
	market to report true preferences to the mechanism. Thus, the
	mechanism successfully aggregates the information in the market and
	levels the playing field for both naive and sophisticated participants.
	
	In some applications, the market does not only match agents, but
	determines also the contractual details of the match. In a labor
	market, firms and workers may have some discretion on how to set the
	salary.  In a
	college admission problem, students can be admitted with or without
	a stipend. In the cadet-to-branch match~\cite[]{SonmezSwitzer2013},
	cadets can choose between different lengths of service time in
	exchange for a higher priority in their branch of choice. These markets can be understood as hybrids between
	matching markets and auctions and have first been analyzed in the
	seminal paper of~\cite{KelsoCrawford1982}, with later important
	extensions
	by~\cite{Roth1984b,Fleiner2003,HatfieldMilgrom2005,HatfieldKojima2009}
	among others. The central mechanism design result in this context states that in a model of many-to-one matching with
	endogenous contracting, a generalized version of the deferred-acceptance mechanism can be defined, and is stable and
	group-strategy-proof under the assumption that contracts are
	substitutes for colleges\footnote{In the following we call the
		applying side of the market ``students" and the admitting side of
		the market ``colleges" motivated by the application of college
		admission. However, the model equally applies to the other
		applications mentioned in this introduction.} and the law of
	aggregate demand holds.\footnote{This means
		that if we expend the choice set of a college, an equal or larger number of contracts is chosen.} Moreover, the deferred acceptance algorithm is outcome equivalent to the cumulative offer algorithm which is a natural modification of the deferred acceptance algorithm. In the cumulative offer algorithm not only acceptances, but also rejections are deferred; in each round of the algorithm, colleges are allowed to choose among all applications they have received so far including those that they have previously rejected.

	In this paper, we study stable and (group)-strategy-proof mechanisms for matching markets with contracts in which students have monotone preferences in contract-terms,\footnote{Colleges' preferences are not necessarily
		monotone in our analysis.} and explore how much the assumptions of substitutability and the law of aggregate demand can be relaxed in this context.  Our research question is motivated by the observation that in many applications, there is a natural ordering over contracts terms, and it is reasonable to assume that preferences are monotone with respect to the ordering:  In college admission problems with
	stipends~\cite[]{Hassidim2017,AbizadaDur2017}, it is natural to assume that students prefer being admitted with a
	stipend to being admitted without a stipend at the same college, or
	more generally, being admitted with a higher stipend than a lower
	stipend at the same college.\footnote{There is empirical evidence
		that monotonicity is violated for reported preferences of some
		participants in the Israeli psychology match, see~\citealp{hassidim2016}. But it seems
		likely that these monotonicity violations can be attributed to students failing to play the
		weakly dominant strategy of revealing preferences truthfully, rather
		than to them having non-monotone preferences.} In the
	cadet-to-branch matching model of~\cite{SonmezSwitzer2013} the
	contract-term is the service time in the military and it is natural
	to assume that cadets prefer a
	shorter to a longer service in the same branch. This assumption is for example made in the analysis
	of~\cite{Jagadeesan2016}.
	For the medical match, mechanisms which allow for flexible salaries
	have been proposed~\cite[]{Crawford2008}, and it seems reasonable
	that doctors in these mechanisms would prefer working for a higher
	salary rather than a lower salary at the same hospital.
	
	We study conditions on choice functions of colleges under which a stable and (group)-strategy-proof matching mechanism exists for our model. Since we work with monotone preferences for students, sufficient conditions on choice functions of colleges are generally weaker. This is because
	certain preference manipulations are ruled out by our model. A
	student must report monotone preferences. Thus, a student cannot, e.g.,
	rank being admitted to a college without a stipend above being
	admitted to the same college with a stipend.\footnote{Similarly, weaker
		conditions are sufficient to guarantee the existence of stable
		allocations than those for markets with non-monotone
		preferences. To illustrate this point, consider a college
		admission problem of the following kind: There are two colleges
		$c_1,c_2$ and three students $s_1,s_2,s_3$. Suppose there are two
		kinds of contracts: A student can be admitted with stipend
		(represented by the contract term "$1$") or without stipend
		(represented by the contract term "$0$"). The colleges have choice
		functions induced (in the usual way) by the following preferences
		\begin{align*}\{(s_1,0),(s_2,0),(s_3,0)\}\succ_{c_1}\{(s_2,0)\}\succ_{c_1}\{(s_2,1)\}\succ_{c_1}\emptyset\succ_{c_1}\ldots\\
		\{(s_1,0)\}\succ_{c_2}\{(s_3,0)\}\succ_{c_2}\emptyset\succ_{c_2}\ldots\end{align*}
		Suppose student always prefer to be admitted at a college under a
		stipend to being admitted at the same college without a stipend.
		Going through all different cases, one can show that, for any
		preferences satisfying this monotonicity assumption, a stable
		allocation (in the matching with contracts sense) exists. This
		changes if students can report non-monotonic preferences. Consider
		the following preferences:
		\begin{align*}%\(f_1,2)\succ_{s_2}(f_2,2)\succ_{s_2}
		(c_1,1)\succ_{s_1}(c_2,1) \succ_{s_1}(c_1,0)\succ_{s_1}(c_2,0)\succ_{s_1}\emptyset\\ %(f_2,2)\succ_{s_3}(f_1,2)\succ_{s_3}
		(c_1,0)\succ_{s_2}(c_1,1)\succ_{s_2}\emptyset\succ_{s_2}\ldots\\
		(c_2,1)\succ_{s_3}(c_1,1)\succ_{s_3}(c_2,0)\succ_{s_3}(c_1,0)\succ_{s_3}\emptyset\end{align*}
		Student $s_2$ prefers to go to college $c_1$ without a stipend
		rather than a stipend. Thus, in a stable allocation it will never be
		the case that $s_2$ goes to college $c_1$ with a stipend, because
		otherwise $c_1$ and $s_2$ could block that allocation. This in turn
		implies that no stable allocation exists: The allocation that
		matches all three students to $c_1$ without stipend is blocked by
		student $s_3$ and college $c_2$. Any allocation that matches $s_1$
		to $c_2$ without stipend is blocked by students $s_1,s_2$ and $s_3$
		being admitted to college $c_1$ without stipend. Any allocation that
		matches $s_3$ to $c_2$ without stipend is blocked by student $s_1$
		being admitted to college $c_2$ without a stipend. Finally, all
		other allocations are either not individually rational or blocked by
		students $s_1,s_2$ and $s_3$ being admitted to college $c_1$ without
		a stipend.} Our main result is a characterization of a sufficient and necessary (in the maximal domain sense) condition for the existence of a stable and (group)-strategy-proof mechanism in terms of ``equivalence" of choice functions: Two choice functions are equivalent if the cumulative offer algorithm yields the same outcome under either choice function, for any preferences and choice functions for the other agents. If for each college an equivalent choice function exists under which contracts are substitutes and the law of aggregate demand holds, then the mechanism that assigns to each profile of preferences the outcome of the cumulative offer algorithm, or equivalently the outcome of the deferred acceptance algorithm, is group-strategy-proof, both under the equivalent choice profile and, by definition of equivalence, under the original choice profile.

	We proceed as follows: First we show (Theorems~\ref{smoothing} and~\ref{andersrum}) that an equivalent choice function under which contracts are substitutes and the law of aggregate demand holds exists, if and only if under the original choice function contracts are observable substitutes (in the sense of~\citealp{Hatfield2018}) and the observable law of aggregate demand holds.\footnote{Importantly, observability is meant with respect to offer sequences consistent with monotone preferences for students. Thus, the notions of observable substitutability and the observable law of aggregate demand are weaker than the corresponding notions for the general model with possibly non-monotone preferences.}  The result crucially relies on our monotonicity assumption on preferences: In independent work,~\cite{Jagadeesan2016} has considered the corresponding notion of equivalence between choice functions for the general preference domain. He shows that equivalence in that context implies unilateral substitutability which is a much more stringent restriction than observable substitutability. Even if the notion of equivalence is relaxed, by allowing, as in the construction of substitutable completions~\cite[]{Hatfield2015}, for non-feasible equivalent choice functions, the corresponding result does not hold on the general preference domain~\cite[]{Hatfield2018}.
	
	Our proof is constructive, by providing a natural construction of an equivalent choice function that we call a "virtual choice function", and by showing that the virtual choice function satisfies substitutability and the law of aggregate demand if and only if the original choice function satisfies observable substitutes and the observable law of aggregate demand. Importantly, the constructed virtual choice function is monotone in contract terms, even if the original choice function is non-monotone in contract terms.
	
	Theorems~\ref{smoothing} and~\ref{andersrum} have several important consequences that are subsequently spelled out: First, Theorem~\ref{smoothing} implies (Corollory~\ref{DAgood}) that the deferred acceptance mechanisms is group-strategy-proof if contracts are observable substitutes and the observable law of aggregate demand holds. Second, combining the theorems with a result of~\cite{Hatfield2018}, we obtain a maximal domain result (Corollary~\ref{equiv}): 
	The domain of choice profiles that are equivalent to a choice profile under which contracts are substitutes for colleges and the law of aggregate demand holds, form a
	maximal Cartesian and unital\footnote{A domain of choice profiles is unital if it contains all profile in which each college has a unit demand choice function, i.e.~it chooses at most one contract from each set.} domain for the existence of a stable
	and (group)-strategy-proof mechanism. Equivalently this domain is the domain of profiles under which contracts are observable substitutes and the observable law of aggregate demand holds. Our results are stronger than the result obtained by~\cite{Hatfield2018} for monotone preferences: For monotone
	preferences,  strategy-proofness of the deferred acceptance mechanism can be strengthened to group-strategy-proofness, and the axiom of ``non-manipulability via contractual terms" in the characterization of the maximal domain becomes redundant.
	 Furthermore, the notion of equivalence yields an easier and more interpretable condition on choice functions: In particular, group-strategy-proof of the deferred acceptance mechanism on the maximal domain now follows immediately from well-known previous results due to~\cite{HatfieldMilgrom2005} and \cite{HatfieldKojima2009} applied to the equivalent choice profile. Third, since an equivalent choice profile can be constructed and the equivalent choice profile is monotone in contract terms (Theorem~\ref{smoothing}), we obtain an
	embedding result (Corollary~\ref{embed}) in the sense of~\cite{Jagadeesan2016} for the class
	of choice functions under which contracts are observable substitutes
	and the observable law of aggregate demand holds. Thus, if attention
	is restricted to the case of monotone preferences for students, it
	is, in some sense without loss of generality to work with the model
	of matching with salaries of~\cite{KelsoCrawford1982} rather than
	the full matching with contracts model.

	\subsection{Related Literature}
	Stable many-to-one matching mechanisms and their incentive
	properties have been extensively
	studied~\cite[]{HatfieldKojima2010,Hatfield2015,Chenetal2014,KominersSonmez2015,
		HirataKasuya2015,Hatfield2018}. Most papers focus
	on the pure matching model or on the matching with contracts
	model. Working
	with monotone preferences makes our model less general than the full
	matching with contracts model. In particular, sufficient
	conditions for stability and the existence of a stable and
	(group)-strategy-proof mechanism from the literature on matching
	with contracts also apply to our model. However, conditions that are
	necessary for the general model with contracts are not necessary
	conditions for the model with monotone preferences. Thus, our results
	are independent of previous results for the matching with
	contracts model. 
	
	In recent related work,~\cite{AbizadaDur2017} consider a model of college admissions with stipends where
	complementarities in  contract terms are present for colleges: In
	their model three contract terms $\{t_+,t_0,t_-\}$ are available,
	interpreted as admission with stipend, admission without stipend but
	with tuition waiver, and admission without either of the two, and
	the number of $t_-$-contracts signed by a college constraints the
	number of $t_+$-contracts it signs. Importantly, students have
	monotone preferences in these contract terms. The model is a special
	case of ours. In particular, the result for ``Max-Min Responsive"
	preferences can be obtained as a special case of ours and for this
	case, their strategy-proofness result can be strengthened to
	group-strategy-proofness. However, \cite{AbizadaDur2017} also
	analyze pairwise-stable outcomes and this part of their analysis
	does not have a counter-part in our paper. 
	
	Our original working paper~\cite[]{Schlegel2016} contained a version
	of our maximal domain result for a model of matching with contracts
	where also colleges' choice functions are monotone in
	contract-terms. Technically the two maximal domain results are
	independent. However, the adaption to obtain the previous version of
	the theorem from the current one are minimal.\footnote{One has to
		make sure that in the ``necessity part" of the maximal domain proof,
		the profile of unit demand choice functions for the other colleges can
		be chosen to be induced by monotone preferences.} The current version of the maximal domain result can also be obtained with a similar proof as the one in the previous version.  For the original version of the theorem, we refer
	the interested reader to the original version of the working paper.
	
	After, our original working paper~\cite[]{Schlegel2016},
	\cite{Hatfield2018} released a new version of their paper,
	where the authors extend their analysis to the case of restricted preference domains for the applicant side and show that their analysis of observable substitutes
	and the observable law of aggregate demand extends to restricted
	preference domains if the notion of ``observability" is adjusted to the domain restriction. Their additional
	work allows us to shorten the proof of our maximal domain result, since we can now make use of their Theorem~5. As remarked previously, our results in the current paper for
	monotone preferences strengthen the results
	in~\cite{Hatfield2018} by strengthening strategy-proofness to group-strategy-proofness.
	\section{Model and Known Results}\label{Model0}
	\subsection{Model}
	There are two finite disjoint sets of agents, a set of {\bf colleges
	}$C$ and a set of {\bf students} $S$. There is a finite set of
	possible {\bf contract-terms} $ T$ which are totally ordered by $\vartriangleright$. Contract terms could, for example, be different amounts of stipends that a college can pay a student, ordered from high to low stipends; in a labor market, contract terms could be different possible wages ordered from high to low, or working hours from short to long, etc.\footnote{More generally, we can deal with the case where the ordering is college-student pair specific, i.e.~for each college $c\in C$ and student $s\in S$ there is a total ordering $\vartriangleright_{c,s}$ over contract terms. If "Pareto separability"~\cite[]{Roth1984b} of contracts holds, the set $T$ can be interpreted as the set of "generalized salaries" in the sense of~\cite{Roth1984b} and $\vartriangleright_{c,s}$ is derived from a parameterization of the efficient contract frontier. Thus, also in this case, our results apply.} Colleges can accept students under
	different bilateral contracts. The set of possible {\bf contracts}
	is ${ X}\subseteq C\times S\times  T$. For a contract $x\in  X$, we
	denote by $x_C\in C$ the college involved in $x$, by $x_S\in S$ the
	student involved in $ X$, and by $x_ T\in  T$ the contract term
	involved in $x$. We write $x_{ T}\unrhd x_{ T}'$ whenever $x_{ T}\vartriangleright x_{ T}'$ or $x=x'$.
	For $Y\subseteq X$ and college $c\in C$ we denote by $Y_c=\{x\in Y: x_C=c\}$ the set of contracts in Y involving $c$, and similarly for a student $s\in S$ we denote by $Y_s:=\{x\in Y:y_S=s\}$ the set of contracts in $Y$ involving $s$.
	
	Each college $c$ has a choice function
	$Ch_c:2^{ X_c}\rightarrow 2^{ X_c}$ that from each set  $Y\subseteq
	X_c$ chooses a subset of contracts. Each college can only sign one
	contract with any given student, i.e.~for each $x,y\in Ch_c(Y)$ with
	$x\neq y$ we have $x_S\neq y_S$. Throughout this paper, we assume that all considered choice functions satisfy the {\bf
		irrelevance of rejected contracts (IRC)}~\cite[]{AygunSonmez2012a}, which means that for all $Y\subseteq  X_c$, $x\in  X_c\setminus Y$, $$x\notin
	Ch_c(Y\cup\{x\})\Rightarrow Ch_c(Y)=Ch_c(Y\cup\{x\}).$$ We also define
	a rejection function $R_c:2^{ X_c}\rightarrow2^{ X_c}$ by
	$R_c(Y):=Y\setminus Ch_c(Y).$ We denote the set of all choice functions for college $c\in C$ that satisfy IRC, by $\mathcal{C}_c$.
	
	% We define the
	%{\bf rejection function} $R_c:2^ X\rightarrow2^ X$ by
	%$R_c(Y):= X\setminus Ch_c(Y)$.
	Each student $s$ has preferences
	$\succeq_s$ over different contracts involving him, and an outside
	option which we denote by ``$\emptyset$". We make the following
	assumption on students' preferences:
	\begin{enumerate}
		\item Preferences are {\bf strict}, for $x,x'\in  X$ with $x'_S=x_S$, we have $$x\neq x'\Rightarrow x\succ_{x_S}x'\text{ or }x'\succ_{x_S}x,$$ and $$x\succ_{x_S}\emptyset\text{  or  }\emptyset\succ_{x_S}x.$$
		
		\item[2.] Preferences are {\bf monotone} in contract terms, for each $x,x'\in  X$ with $x_S=x'_S$ and $x_C=x_C'$ we have
		$$x_ T\vartriangleright x'_ T\Rightarrow x\succ_{x_S}x'.$$
	\end{enumerate}
	We denote the set of all strict and monotone preferences for student $s\in S$ by $\mathcal{R}_s$.
	A {\bf problem} is a pair
	$(Ch,\succeq)$ consisting of a {\bf choice profile} $Ch=(Ch_c)_{c\in
		C}\in\bigtimes_{c\in C}\mathcal{C}_c$ and a {\bf preference profile} $\succeq=(\succeq_s)_{s\in S}\in\bigtimes_{s\in S}\mathcal{R}_s$.

	An {\bf allocation} is a set $Y\subseteq  X$ that contains at most
	one contract per student. We denote the set of allocations by
	$\mathcal{A}$. For $Y\in\mathcal{A}$, we write $Y_S:=\{y_S:y\in Y\}$. An allocation $Y$ is
	\begin{itemize}
		\item[] {\bf individually rational} in $(Ch,\succeq)$ if for each $c\in C$, we have $Y_c= Ch_c(Y)$, and for each $y\in Y$ we have $y\succ_{y_{ S}}\emptyset,$
		\item[] {\bf blocked} in $(Ch,\succeq)$ if there are $c\in C$ and an allocation $Z$ with $Z_c\neq Y_c$, such that $Z_c= Ch_c(Y\cup Z)$ and for each $z\in Z_c$ we have $z\succeq_{z_{S}}Y_{z_{S}}$,
		\item[] {\bf stable} in $(Ch,\succeq)$ if it is individually rational and not blocked.
	\end{itemize}
 In the following, it will be useful to define for each
	set of contracts $Y\subseteq  X_c$ with a college $c$ the allocation
	$$Y^{\min}:=\{y\in Y:\nexists y'\in Y, y_S'=y_S,y_
	T\vartriangleright y'_ T\}$$ of contracts that gives each student
	the worst contract among the contract in $Y$, and the set
	$$\mathcal{U}(Y):=\{x\in  X_c:x_{ T}\unrhd y_ T\text{ for
		some }y\in Y\text{ with }x_S=y_S\}$$ of contracts,
	not necessarily in $Y$, which are as least as good for the involved student as his worst contract in $Y$.\footnote{Note that the definitions of $Y^{\min}$ and $\mathcal{U}(Y)$ only depend on $\vartriangleright$.} 
	We call choice function 
	$Ch_c$ for college $c\in C$ {\bf monotone in contract-terms} if for each $Y\subseteq
	X_c$ we have
	$$Ch_c(Y)\subseteq Y^{\min}.$$
	
	\subsection{Mechanisms}
	A {\bf mechanism} (for the students) is a mapping from preference
	profiles to allocations
	$\mathcal{M}:\bigtimes_{s\in S}\mathcal{R}_s\rightarrow\mathcal{A}$. Mechanism
	$\mathcal{M}$ is {\bf strategy-proof} if it is a weakly dominant
	strategy for students to report their true preferences to the
	mechanism, i.e.~for each $s\in S$,
	$\succeq_{-s}\in\bigtimes_{s'\in S\setminus\{s\}}\mathcal{R}_{s'}$ and
	$\succeq_s,\succeq'_s\in\mathcal{R}_s$ we have
	$$\mathcal{M}(\succeq_{s},\succeq_{-s})\succeq_s\mathcal{M}(\succeq'_s,\succeq_{-s}).$$
	Mechanism $\mathcal{M}$ is {\bf group-strategy-proof} if for each
	$S'\subseteq S$, $\succeq_{-S'}\in\bigtimes_{s\in S\setminus\mathcal{S'}}\mathcal{R}_s$ and
	$\succeq_{S'},\succeq'_{S'}\in\bigtimes_{s\in S'}\mathcal{R}_s$, there is a $s'\in
	S'$ with
	$$\mathcal{M}(\succeq_{S'},\succeq_{-S'})\succeq_{s'}\mathcal{M}(\succeq'_{S'},\succeq_{-S'}).$$
	Let $Ch$ be a choice profile. Mechanism $\mathcal{M}$ is {\bf
		$Ch$-stable} if for each $\succeq\in\bigtimes_{s\in S}\mathcal{R}_s$ allocation
	$\mathcal{M}(\succeq)$ is stable in $(Ch,\succeq)$.
	\subsection{Examples}
	Several examples from applied marked design fit into our model, such as the model of cadet-to-branch matching of~\cite{Sonmez2013} with the preference restriction of~\cite{Jagadeesan2016}, the discrete version of the job matching model of~\cite{KelsoCrawford1982}, the college admission models of~\cite{Hassidim2017} and of~\cite{AbizadaDur2017}.  The job matching model of~\cite{KelsoCrawford1982} will play an important role in Section~\ref{sec:embed}. 
	\subsubsection{Job matching with salaries}
	A finite {\bf Kelso-Crawford economy} consists of a finite set of
	firms $F$, a finite set of workers $W$, a finite set of salaries
	$\Sigma\subseteq \mathbb{R}_{+}$ and a profile $u_{i\in F\cup W}$
	of utility functions, where for each $f\in F$, utility is a functions
	$u_f:\{(W',p):W'\subseteq W,p\in\mathbb{R}^{W'}_{+}\}\rightarrow\mathbb{R}$ that is strictly decreasing in $\mathbb{R}_+^{W'}$ for each $W'$, and for each
	$w\in W$, utility is a function $u_w:F\times\mathbb{R}_{+}\cup\{\emptyset\}\rightarrow\mathbb{R}$ that is strictly increasing in $\mathbb{R}_+$ for each $f$.
	
	The model fits in our framework with $C=F$, $S=W$, $T=\Sigma$,
	$x_T\vartriangleright x_T'\Leftrightarrow x_T>x_T'$, $X=F\times
	W\times \Sigma$, choice functions are defined by
	$$Ch_c(Y)=\max_{Y'\subseteq Y^{\min}} u_c(Y'_{S},(y_T)_{y\in Y'}),$$
	and preferences $(\succeq_s)_{s\in S}$ are induced by utility
	functions, $$(c,s,t)\succ_s(c',s,t')\Leftrightarrow u_s(c,t)>u_s(c',t'),\quad (c,s,t)\succ_s\emptyset\Leftrightarrow u_s(c,t)>u_s(\emptyset).$$
	Note that the constructed market with contracts $(Ch,\succeq)$ does
	not only satisfy monotonicity of students' preferences, but also
	monotonicity of colleges' choice function.
	We say that workers are {\bf gross substitutes} for firm $f$ if for each $p,p'\in\Sigma^W$ such that $p\leq p'$, if $p_w'=p_w$ and $w\in\text{argmax}_{W'\subseteq W}u_f(W',p)$, then $w\in\text{argmax}_{W'\subseteq W}u_f(W',p')$.

	\subsection{Stable Allocations}\label{exist}
	In general, a stable allocation does not need to exist for our
	model. A sufficient condition for stability is that contracts are
	substitutes for colleges, i.e.~if a contract is rejected from some set
	of contracts, then this contract is also rejected from each superset of
	that set of contracts.
	\newline\newline
	\noindent {\bf Substitutability}
	\citep{Roth1984b,HatfieldMilgrom2005}: For each $Z\subseteq Y\subseteq  X_c$, $$R_c(Z)\subseteq R_c(Y).$$
	
	Not only is substitutability sufficient for the existence of a
	stable allocation but it also guarantees that the set of stable
	allocations has a lattice structure. If contracts are substitutes
	for colleges, then the set of stable allocation forms a lattice with
	respect to the preferences of students~\cite[]{Blair1988}. In
	particular, there is a unique stable allocation that is most
	preferred by all students among all stable allocations. We call this
	allocation the {\bf student-optimal stable} allocation.  It can be
	found by the {\bf cumulative offer (CO) algorithm} that is
	defined as follows.
	\begin{enumerate}
		\item Each student applies under his favorite acceptable and unrejected contract or stays alone if he finds no unrejected contract acceptable.
		\item Each college tentatively accepts the contracts it chooses among the applications it has received so far and rejects all other contracts.
		\item If no application has been rejected in the current round, the algorithm stops. If some applications are rejected we repeat.
	\end{enumerate}
	We denote the set of contracts chosen by the colleges in the final round of the CO algorithm for problem $(Ch,\succeq)$ by $\mathcal{CO}(Ch,\succeq)$. Alternatively, we can use the deferred acceptance algorithm to find the student optimal stable allocation. In the {\bf deferred acceptance (DA) algorithm} rejections are permanent: Thus, in 2., colleges only choose among contracts that are proposed in the current round of the algorithm, but not contracts that have been proposed and rejected in a previous round. If contracts are substitutes for colleges, then the deferred acceptance and the cumulative offer algorithm yield the same outcome. The {\bf deferred acceptance mechanism} for $Ch$ assigns to each $\succeq\in\bigtimes_{s\in S}\mathcal{R}_s$ the outcome of the deferred acceptance algorithm in $(Ch,\succeq)$.\footnote{We use the deferred acceptance algorithm and not the cumulative offer algorithm to define a mechanism, because it yields a well-defined mechanism even in situations where contracts are not substitutes for colleges. The deferred acceptance algorithm always ends in a (potentially unstable) allocation where each student is accepted at at most one college, whereas the cumulative offer algorithm can yield an outcome where multiple colleges accept a student.}

	Student-optimality is related to group-strategy-proofness. Under
	substitutability and the following additional condition on the
	colleges' choice functions the deferred acceptance mechanism is
	group-strategy-proof.
	\newline

	\noindent
	{\bf Law of Aggregate Demand} \citep{HatfieldMilgrom2005}:  For each $Z\subseteq Y\subseteq  X_c$: $$Z\subseteq Y\Rightarrow |Ch_c(Y)|\geq|Ch_c(Z)|.$$
	
	The following proposition summarizes known results about
	side-optimal stable allocations, the invariance of the set of
	matched students in stable allocations (the rural hospitals
	theorem), and group-strategy-proofness of the deferred acceptance mechanism.
	\begin{proposition}[\citealp{KelsoCrawford1982,Blair1988,Fleiner2003,HatfieldMilgrom2005,HatfieldKojima2009}]\label{SufficientWeak1}
		\noindent
		\begin{enumerate}
			\item  If contracts are substitutes for colleges, then the deferred acceptance and the cumulative offer algorithm are equivalent and converge to a stable allocation that is most preferred by all students among all stable allocations.
			\item If choice functions satisfy, moreover, the law of aggregate demand, then \begin{enumerate}
				\item the set of accepted students is the same in all stable allocations and each college accepts the same number of students in each stable allocation,
				\item the deferred acceptance mechanism is group-strategy-proof.
			\end{enumerate}
		\end{enumerate}
	\end{proposition}
	\section{Results}
	\subsection{Equivalent Choice Functions}\label{virtual}
	It is a natural question whether the conditions of
	Section~\ref{exist} for the stability and group-strategy-proofness
	of the deferred acceptance mechanism are also necessary. Next we
	provide a counterexample that shows that substitutability and the law
	of aggregate demand are not necessary for the deferred-acceptance
	mechanism to be stable and group-strategy-proof. The example will have the following structure: There is one college $c$
	for which contracts are not substitutes. For each college except for
	$c$, contracts are substitutes and the law of aggregate demand
	holds. However, $c$'s choice function can be replaced by another
	choice function, such that \begin{enumerate}
		\item under the replacing choice function, contracts are
		substitutes for $c$ and the law of aggregate demand
		holds,
		\item the outcome of
		the cumulative offer algorithm (which here coincides with the outcome of the deferred acceptance algorithm) is, for all monotone preference profiles for students, the same under the original
		choice profile and the profile where $c$'s choice function is
		replaced.\end{enumerate} As the deferred acceptance mechanism is group-strategy-proof for the profile
	where we have replaced $c$'s choice function by the other choice
	function, it is by outcome-equivalence also group-strategy-proof for the original choice profile.
	\begin{example}\label{virt}
		The original choice function is from Example~2 in~\cite{Hassidim2017}. Consider a college admission problem with stipends where $T=\{0,1\}$, with the interpretation that contract term "$0$" corresponds to admission without stipend, and contract term $"1"$ corresponds to admission with a stipend. Thus, $1\vartriangleright0$. For simplicity, we assume that there are two students $S=\{s_1,s_2\}$. The example can be generalized to more students.
		Consider a college $c_1$ which can admit up to two students but only has one stipend available. Students are admitted based on merit, and student $s_1$ is the higher merit student than $s_2$. If both students are available, the college prefers to allocate the stipend to student $s_2$ rather than $s_1$, e.g. because $s_2$ comes from a lower income background. The choice function of the college $Ch_{c_1}$ is induced by preferences 
		\begin{align*}\{(c_1,s_1,0),(c_1,s_2,0)\}\succ_{c_1}\{(c_1,s_1,0),(c_1,s_2,1)\}\succ_{c_1}\{(c_1,s_1,1),(c_1,s_2,0)\}\\\succ_{c_1}\{(c_1,s_1,0)\}\succ_{c_1}\{(c_1,s_1,1)\}\succ_{c_1}\{(c_	,s_2,0)\}\succ_{c_1}\{(c_1,s_2,1)\}
		\end{align*}
		in the usual way. 
		Now consider the alternative choice function $Ch_{c_1}'$ induced by preferences:
		\begin{align*}\{(c_1,s_1,0),(c_1,s_2,0)\}\succ_{c_1}'\{(c_1,s_1,1),(c_1,s_2,0)\}\succ_{c_1}'\{(c_1,s_1,0)\}\\\succ'_{c_1}\{(c_1,s_1,1)\}\succ_{c_1}'\{(c_1,s_2,0)\}\succ_{c_1}'\{(c_1,s_2,1)\}
		\end{align*}
		Thus, the preferences differ only in so far as allocation $\{(c_1,s_1,0),(c_1,s_2,1)\}$ where $s_1$ is admitted without stipend and $s_2$ is admitted with stipend is unacceptable under $\succeq_{c_1}'$. The ranking of all other allocations is the same under both preferences.
		
		Note that under $Ch_{c_1}$ contracts are not substitutes as  $$(c_1,s_2,1)\in
		Ch_{c_1}(\{(c_1,s_1,0),(c_1,s_1,1),(c_1,s_2,1)\})=\{(c_1,s_1,0),(c_1,s_2,1)\}$$ but $$(c_1,s_2,1)\notin Ch_{c_1}(\{(c_1,s_1,1),(c_1,s_2,1)\})=\{(c_1,s_1,1)\},$$ and that
		under $Ch_{c_1}'$ contracts are substitutes and the law of aggregate demand holds.
		
		Suppose colleges $C\setminus \{c_1\}$ have choice functions
		$Ch_{-c_1}=(Ch_{c})_{c\neq c_1}$ under which contracts are substitutes
		and the law of aggregate demand holds. Define $Ch:=(Ch_{c_1},Ch_{-c_1})$ and
		$Ch':=(Ch_{c_1}',Ch_{-c_1})$. Let $\succeq\in\bigtimes_{s\in S}\mathcal{R}_s$. We show that the cumulative offer
		algorithm  converges to the
		same allocation for the problems $(Ch,\succeq)$ and  $(Ch',\succeq)$. Observe that for all sets $Y\subseteq X$ with 
		$Ch_{c_1} (Y)\neq Ch_{c_1}'(Y)$ we have $\{(c_1,s_1,0),(c_1,s_2,1)\}\subseteq Y$. In
		particular, for the cumulative offer algorithm to differ for the
		two problems, student $s_1 $ must apply to $c_1$ for admission without stipend
		during the cumulative offer algorithm in $(Ch,\succeq)$. Note
		however that before $s_1$ applies to $c_1$ for admission without stipend, he
		applies to $c_1$ for admission with stipend, as $1\vartriangleright0$.
		But once college $c_1$ tentatively accepts $s_1$ with a stipend,
		the college will not subsequently drop student $s_1$, as $s_1$ has higher merit than $s_2$. Thus, $s_1$ will never apply to $c_1$ for admission without a stipend during the cumulative offer algorithm. Hence, the cumulative offer algorithm for the two problems converges to the same allocation, which
		is the student-optimal stable allocation in $(Ch',\succeq)$.
		
		Note, moreover, that the equivalent choice function $Ch'_c$ prioritizes merit over need in the following sense: From each set of contracts, an allocation is chosen, where the stipend is either un-allocated or allocated to the higher merit student $s_1$; whereas under choice function, $Ch_c$ the stipend would be allocated to student $s_2$, if both students are available and it is possible to admit student $s_1$ without a stipend.
		\qed
	\end{example}
	The example motivates the following definition.\footnote{A similar
		notion has been introduced independently by~\cite{Jagadeesan2016}.
		However, since \cite{Jagadeesan2016} does not assume monotonicity on
		students' preferences, his notion of equivalence is much stronger.
		See our Remark~1.}
	\newline\newline
	\noindent
	{\bf CO-equivalence}:
	Two choice profiles $Ch$ and $Ch'$ are
	CO-equivalent if for each  $\succeq\in\bigtimes_{s\in S}\mathcal{R}_s$ we have $\mathcal{CO}(Ch,\succeq)=\mathcal{CO}(Ch',\succeq)$.
	Two choice functions $Ch_c$ and $Ch_c'$ are CO-equivalent if for each $Ch_{-c}\in\bigtimes_{c'\in C\setminus\{c\}}\mathcal{C}_{c'}$ profiles $Ch$ and $Ch':=(Ch'_c,Ch_{-c})$ are CO-equivalent.
	\newline
	
	In the following, we use the term "equivalent" instead of "CO-equivalent", since we only consider one notion of equivalence throughout the paper. 
	From the definition and Proposition~\ref{SufficientWeak1}, we obtain the following result:
	\begin{proposition}\label{GSP}
		If $Ch$ is equivalent to a profile under which contracts are
		substitutes and the law of aggregate demand holds, then the deferred
		acceptance mechanism for $Ch$ is $Ch$-stable and group-strategy-proof.
	\end{proposition}
	\begin{proof}
		Let $Ch'$ be the equivalent profile. Let $\succeq\in\bigtimes_{s\in S}\mathcal{R}_s$.
		Let $Y_1\subseteq Y_2\subseteq \ldots\subseteq Y_{\tau}$ be the set of accumulated offers during the cumulative offer algorithm for $(Ch',\succeq)$ and $Z_1,Z_2,\ldots, Z_\tau$ be the set of proposed contracts during the deferred acceptance algorithm for $(Ch',\succeq)$. Since contracts under $Ch'$ are substitutes, Part~(1) of  Proposition~\ref{SufficientWeak1} implies $Ch'_c(Y_t)=Ch'_c(Z_t)$ for $t=1,\ldots,\tau$ and each $c\in C$. Since $Ch$ and $Ch'$ are equivalent we have $Ch_c(Y_t)=Ch'_c(Y_t)$ for $t=1,\ldots,\tau$ and each $c\in C$. Thus, $Ch_c(Y_t)=Ch'_c(Y_t)=Ch'_c(Z_t)\subseteq Z_t$ for $t=1,\ldots,\tau$ and each $c\in C$. By IRC for $Ch$ this implies $Ch_c(Z_t)=Ch_c(Y_t)=Ch'_c(Z_t)$ for $t=1,\ldots,\tau$ and each $c\in C$, and the deferred acceptance algorithm in $(Ch,\succeq)$ and in $(Ch',\succeq)$ yield the same outcome. In particular, this implies the stability in $(Ch,\succeq)$ of the outcome $\bigcup_{c\in C}Ch_c(Z_{\tau})$ of the deferred acceptance algorithm, since otherwise, if college $c\in C$ blocks, we would have $Ch_c(Y_{\tau})\nsubseteq Z_{\tau}.$ Moreover, if under $Ch'$ contracts are substitutes and the law of aggregate demand holds, this implies by Part 2(b) of Proposition~\ref{SufficientWeak1} the group-strategy-proofness of the deferred acceptance mechanism for $Ch$.
	\end{proof}
	As our first main result, we show that for each choice function under which contracts are observable substitutes and the observable law of aggregate demand holds
	we can construct an equivalent choice function under which contracts are substitutes and the law of aggregate demand holds. As an immediate corollary, we will obtain the group-strategy-proofness of the deferred acceptance mechanism under observable substitutes and the observable law of aggregate demand.
	Importantly, we define observable substitutability and the observable law of aggregate demand for monotone preferences by only requiring that the properties hold for observable sequences that can occur under monotone preferences:
	In the following, a sequence of contracts $x^1,x^2,\ldots,x^{\tau}$ is {\bf generated
		from monotone preferences} if for $1\leq t\leq
	\tau$ and each $x\in  X$ with $x_S=x_S^t$ and
	$x_ T\vartriangleright x_ T^t$, we have $x\in\{x^1,\ldots,x^{t-1}\}$. A
	sequence $x^1,x^2,\ldots,x^{\tau}$ is {\bf observable} under $Ch_c$ if for $1\leq t\leq
	\tau-1$ we have $x^{t+1}_S\notin [Ch_c\{x^1,\ldots,x^{t}\}]_S$. We call a set of contracts $Y\subseteq  X_c$ {\bf
		observable} under $Ch_c$, if there is a sequence of contracts
	$x^1,\ldots,x^{|Y|}$ that is observable under $Ch_c$ and that is generated
	from monotone preferences such that $Y=\{x^1,\ldots,x^{|Y|}\}.$ We denote the set of all observable sets of contracts under $Ch_c$ by $\mathcal{O}(Ch_c)$. Observable substitutability for monotone preferences is defined as follows:\footnote{Alternatively, we could also define observable substitutabilty for monotone preferences, by requiring that the substitutes condition holds for observable sets, i.e.~by requiring that for $Y,Z\in\mathcal{O}(Ch_c)$ and $Z\subseteq Y$ we have $R_c(Z)\subseteq R_c(Y)$. See the first part of Lemma~\ref{order2} in the appendix. Note that this definition is only equivalent for the case of monotone preferences, but is stronger on the general preference domain.}
	\newline\newline
	\noindent
	{\bf Observable Substitutability (for Monotone Preferences)}~\cite[]{Hatfield2018}:
	Contracts are {observably substitutes
		for monotone preferences} under choice function $Ch_c$ if for each observable sequence
	$x^1,x^2,\ldots,x^{\tau}$ under $Ch_c$ that is generated from
	monotone preferences, we have
	$$R_c\{x^1,\ldots x^{t}\}\subseteq R_c\{x^1,\ldots,x^{t+1}\}\text{ for each }1\leq t<\tau.$$
	Similarly, we can define an observable version of the law of aggregate demand.
	\newline\newline
	\noindent
	{\bf Observable Law of
		Aggregate Demand (for Monotone Preferences)}~\cite[]{Hatfield2018}:
	A choice function $Ch_c$ satisfies the {observable law of
		aggregate demand for monotone preferences} if for each observable
	sequence under $Ch_c$ that is generated from monotone preferences, we
	have
	$$|Ch_c\{x^1,\ldots x^{t}\}|\leq|Ch_c\{x^1,\ldots,x^{t+1}\}|\text{ for each }1\leq t<\tau.$$
	
	Since we exclusively deal with monotone preferences, from now one we
	drop the term "for monotone preferences." However, the reader should be aware that the conditions are weaker than the corresponding conditions for general preferences.\footnote{In the example in Footnote~6, for example, college $c_1$'s choice function satisfies the monotone versions of the two conditions but not the general preference versions.} With the definition, we have the following result.	
	\begin{theorem}\label{smoothing}
		For each choice function $Ch_c$ under which contracts are observable substitutes and the observable law of aggregate demand holds, there exists an equivalent  choice function $Ch'_c$ that is monotone in contract terms such that contracts are substitutes and the law of aggregate demand holds.
	\end{theorem}
	To prove the result, we define for each choice function $Ch_c$ a related choice function $Ch'_c$ that we call the
	{\bf virtual choice function} for $Ch_c$ by $$Ch'_c(Y):=Y^{\min}\setminus\bigcup_{\tilde{Y}\subseteq \mathcal{U}(Y),\tilde{Y}\in\mathcal{O}(Ch_c)}R_c(\tilde{Y}).$$
	The construction has the following intuition: The set of accumulated offers during an instance of the cumulative offer algorithm is observable. Hence a sufficient condition for equivalence between two choice functions $Ch_c$ and $Ch_c'$ is that $Ch_c(Y)=Ch'_c(Y)$ and therefore $R_c(Y)=R_c'(Y)$ for each $Y\in\mathcal{O}(Ch_c)$.  If contracts are substitutes under $Ch'_c$, then, in particular, contracts that are rejected from an observable (under $Ch_c$) set are also rejected from each, potentially unobservable, superset of that set. Hence assuming that choices agree on observable sets and contracts are substitutes under $Ch_c'$ for each $Y\subseteq X_c$ we have  $\bigcup_{\tilde{Y}\subseteq Y,\tilde{Y}\in\mathcal{O}(Ch_c)} R_c(\tilde{Y})=\bigcup_{\tilde{Y}\subseteq Y,\tilde{Y}\in\mathcal{O}(Ch_c)} R'_c(\tilde{Y})\subseteq R'_c(Y).$ Thus, $$Ch'_c(Y)\subseteq Y\setminus\bigcup_{\tilde{Y}\subseteq Y,\tilde{Y}\in\mathcal{O}(Ch_c)}R_c(\tilde{Y}).$$ To guarantee feasibility, $Ch'_c(Y)$ generally has to be a strict subset in the previous expression. Choosing only contracts in $Y^{\min}$ guarantees feasibility and monotonicity of the choice function $Ch_c'$.
	To guarantee the IRC condition for $Ch'_c$ on sets with $Y\neq\mathcal{U}(Y)$ we need in general to consider observable subsets of $\mathcal{U}(Y)$ rather than $Y$. The full proof of the theorem is in the appendix.
	\begin{remark}
		A corresponding result to Theorem~\ref{smoothing} does not hold for the case of non-monotonic preferences.
		Without monotonicity, equivalence requires that choices coincide
		on all sets of contracts which contain at most one contract per
		student (see Theorem~1 of \citealp{Jagadeesan2016}), and equivalence to a choice function under which contracts are substitutes and the law of aggregate demand holds implies unilateral substitutability and the law of aggregate demand. These properties are much stronger than observable substitutability and the observable law of aggregate demand.\qed 
	\end{remark}
\begin{remark}
	One can show that, for the job matching model with salaries, if the choice function in Theorem~\ref{smoothing} is generated from maximizing preferences that are quasi-linear in salaries, then the choice function and the equivalent virtual choice function are the same. Combining this observation with the maximal domain results of~\cite{Hatfield2018}, one can show that for quasi-linear preferences, choice profiles under which contracts are substitutes form a maximal domain of choice profiles generated from quasi-linear preferences for the existence of a stable (and strategy-proof) mechanism. This yields an alternative way of proving a result due to~\citealp{HatfieldMilgrom2005,HatfieldKojima2008}.\qed
\end{remark}
	Immediately from Theorem~\ref{smoothing} and Proposition~\ref{GSP} it follows that the deferred acceptance mechanism is stable and group-strategy-proof if contracts are observable substitutes and the observable law of aggregate demand holds.
	\begin{corollary}\label{DAgood}
		For each choice profile under which contracts are observable substitutes and the observable law of aggregate demand holds, the deferred acceptance mechanism is stable and group-strategy-proof. 
	\end{corollary}
	\begin{remark}\label{w00t}
		\cite{Hatfield2018} prove that observable substitutability, the observable law of aggregate demand and "non-manipulability via contractual terms" are sufficient conditions for the stability and strategy-proofness of the deferred acceptance mechanism. 
		Non-manipulability via contractual terms requires that for each profile in the domain and each college $c\in C$ the deferred-acceptance mechanism is strategy-proof on the domain of
		preference profiles where only contracts with $c$
		are acceptable for students. For the case of monotone preferences, Corollary~\ref{DAgood} strengthens the result in two directions: First, strategy-proofness can be replaced by the stronger property of group-strategy-proofness. Second,	whereas non-manipulability via contractual terms is in general independent of the other two properties, Corollary~\ref{DAgood} implies that this property is implied by observable substitutability and the observable law of aggregate demand for monotone preferences.\qed
	\end{remark}
	The assumption of observable substitutability turns out to be necessary for Theorem~\ref{smoothing}, as we show in the following theorem.
	\begin{theorem}\label{andersrum}
		If $Ch$ is equivalent to a profile $Ch'$ under which contracts are substitutes, then contracts are observable substitutes under $Ch$. If, moreover, the law of aggregate demand holds under $Ch'$, then the observable law of aggregate demand holds under $Ch$.
	\end{theorem}
	\subsection{A Maximal Domain Result}
	Next, we derive several implications of Theorems~\ref{smoothing} and~\ref{andersrum}. Combining the results with Theorem~5 of~\cite{Hatfield2018} we first characterize a maximal domain of choice profiles for the existence of a stable and strategy-proof mechanism using our notion of equivalence. In the following a {\bf
		choice domain} is a set of choice profiles
	$\mathcal{D}\subseteq\bigtimes_{c\in C}\mathcal{C}_{c}$. A choice
	domain $\mathcal{D}$ is {\bf Cartesian} if
	$\mathcal{D}=\bigtimes_{c\in C}\mathcal{D}_c$ where
	$\mathcal{D}_c\subseteq\mathcal{C}_c$ for each college $c\in C$. A
	choice function $Ch_c$ for college $c$ is {\bf unit demand} if
	$|Ch_c(Y)|\leq 1$ for all $Y\subseteq  X_c$. A choice domain is {\bf
		unital} if it includes all profiles of unit demand choice
	functions. We obtain the following corollary of Theorems~\ref{smoothing} and~\ref{andersrum}.
	\begin{corollary}\label{equiv}
		For a domain $\mathcal{D}$ of choice profiles the following statements are equivalent:
		\begin{enumerate}
			\item $\mathcal{D}$ is a maximal Cartesian, unital domain of choice profiles such that a stable and strategy-proof mechanism is guaranteed to exist,
			\item $\mathcal{D}$ is the domain of choice profiles that are equivalent to a choice profile under which contracts are substitutes and the law of aggregate demand holds.
		\end{enumerate}
	\end{corollary}
	\begin{proof}
		
		Let $\mathcal{D}_1$ be a maximal Cartesian, unital domain of choice profiles such that a stable and strategy-proof mechanism is guaranteed to exist. Let $\mathcal{D}_2$ he domain of choice profiles that are equivalent to a choice profile under which contracts are substitutes and the law of aggregate demand holds. By Theorems~\ref{smoothing} and~\ref{andersrum}, $\mathcal{D}_2$ is the domain of choice profiles under which contracts are observable substitutes and the observable law of aggregate demand holds. By Theorem~5 of~\cite{Hatfield2018}, $\mathcal{D}_1$ is a subdomain of the domain of profiles under which contracts are observable substitutes and the observable law of aggregate demand holds, thus $\mathcal{D}_1\subseteq \mathcal{D}_2$. By Proposition~\ref{GSP}, for each $Ch\in\mathcal{D}_2$, the deferred acceptance mechanism is stable and strategy-proof. Thus, by maximality $\mathcal{D}_1=\mathcal{D}_2$.

	\end{proof}
\begin{remark}
	Combining the result with Proposition~\ref{GSP}, the maximal domain result also holds if strategy-proofness is replaced by group-strategy-proofness.\qed
\end{remark}

	\subsection{An embedding result}\label{sec:embed}
	As a third corollary of Theorem~\ref{smoothing}, we obtain an
	embedding result. In recent work~\cite{Jagadeesan2016} shows that
	for each BfYC choice profile as introduced by~\cite{Sonmez2013} there
	is an equivalent choice profile such that the equivalent problem
	can be embedded into a Kelso-Crawford economy. Formally, an {\bf isomorphism}
	between a matching market with contracts $(Ch,\succeq)$ and a
	Kelso-Crawford economy $(\Sigma,u)$ is a one-to-one mapping
	$(f,w,\sigma):X\rightarrow F\times W\times\Sigma$
	such that
	\begin{enumerate}
		\item for each $x,x'\in X$,  $f(x)=f(x')$ if and only if $x'_C=x_C$, and $w(x)=w(x')$ if and only if $x_S=x_S'$,
		\item for each $x,x'\in X$ we have $$x\succ_{x_S}x'\Leftrightarrow u_{w(x)}(f(x),\sigma(x))>u_{w(x)}(f(x),\sigma(x))$$
		and $$x\succ_{x_S}\emptyset\Leftrightarrow
		u_{w(x)}(f(x),\sigma(x))>u_{w(x)}(\emptyset),$$
		\item for each $c\in C$, $Y\subseteq X_c$ and $f\in F$ such that $f(x)=f$ for each $x\in Y$ we have $$Ch_c(Y)=\text{argmax}_{Y'\subseteq Y^{\min}}u_{f}(\{(w(x),\sigma(x)):x \in Y'\}),$$
		\item for each $w\in W,f\in F$ and  $\sigma\in\Sigma\setminus \{\sigma(x):x\in X, f(x)=f,w(x)=w \}$, we have $u_{w}(f,\sigma)<u_{w}(\emptyset)$.\footnote{This item is not required in~\cite{Jagadeesan2016}, as he only considers situations in which there is the same number of contracts between any college and student.}
	\end{enumerate}
\begin{remark}
	 The above notion of an isomorphism is due to~\cite{Jagadeesan2016}. It is different from the notion introduced by~\cite{Echenique2012}, and used by~\cite{Kominers2011}, and \cite{Schlegel2015}: 
			Utility function in the Kelso-Crawford economy satisfy stronger
			regularity conditions in~\cite{Jagadeesan2016}.
			In~\cite{Echenique2012}, monotonicity of utility functions is only required for salaries corresponding to ``un-dominated" contracts
			(see the discussion in~\citealp{Schlegel2015}), whereas in
			\cite{Jagadeesan2016} monotonicity can be achieved for all salaries.
			
			To guarantee that utility functions satisfy the stronger regularity condition,~\cite{Jagadeesan2016} first constructs an equivalent choice profile. The embedding is then performed for a problem where the original choice profile is replaced by the equivalent choice profile. In contrast to this, the embedding result of~\cite{Echenique2012} establishes an isomorphism (with his weaker notion of isomorphism) for the original market with contracts, without the need to first construct an equivalent choice profile. Put differently, the embedding result of~\cite{Jagadeesan2016} establishes an isomorphism between the cumulative offer algorithm in the original problem and the salary adjustment process in the Kelso-Crawford economy, whereas the embedding result of~\cite{Echenique2012} establishes an isomorphism between the sets of stable allocations in the original problem and the Kelso-Crawford economy, but under a weaker notion of an isomorphism.
			
			In the following, we will establish an embedding result in the sense of~\cite{Jagadeesan2016}. An embedding result in the sense of~\cite{Echenique2012} does not, in general, hold in our setting.
\qed
\end{remark}
	\begin{corollary}\label{embed}
		For each choice profile $Ch$ such that contracts are observable
		substitutes and the observable law of aggregate demand holds, there
		exists an equivalent choice profile $Ch'$, such that for each preference profile $\succeq$, the market $(Ch',\succeq)$ is isomorphic to a Kelso-Crawford economy $(\Sigma,u)$ such that workers are gross substitutes for firms.
	\end{corollary}
	\begin{proof}
		By Theorem~\ref{smoothing}, there exist an equivalent profile $Ch'$ under which contracts are substitutes and the law of aggregate demand holds. Moreover, $Ch'$ can be chosen to be monotone in contract-terms for each college. For each $c\in C$, $Ch_c'$ can be rationalized (see Theorem~1 in \citealp{alva2018}) by a strict preference relation $\succeq'_c$ over $\mathcal{A}_c$. The preference relation $\succeq'_c$ can be chosen to be monotonic in contract terms. Let $(\succeq_s)_{s\in S}$ be a preference profile for students. For the preference profile $((\succeq'_c)_{c\in C},(\succeq_s)_{s\in S})$, we can construct as in the proof of Theorem~1 in~\cite{Echenique2012} a corresponding Kelso-Crawford economy $(\Sigma,u)$ such that Gross Substitutability for firms hold. Since $\succeq'_c$ is monotonic in contract terms for each $c\in C$ there are no undominated contracts in $(Ch',\succeq)$. Thus, the embedding is one-to-one and satisfies our definition of an isomorphism, in particular monotonicity of the utility functions in the Kelso-Crawford economy is satisfied.
		\end{proof}
	\begin{remark}
		While we use the same notion of isomorphism as
		~\cite{Jagadeesan2016}, we do not assume quasi-linearity of firm
		utility functions in the Kelso-Crawford economy. In this sense, our
		result is weaker. However, our result applies to a larger domain
		of choice functions.\qed
	\end{remark}
\section{Conclusion}
We have studied a model of matching with contracts with a natural monotonicity condition on applicants' preferences that is satisfied in many practically important matching markets, e.g., for the matching of students to colleges with different stipends or fees, the medical match with flexible salaries, or the cadet-to-branch match if cadets rank branch-of-choice contracts monotonically. Our domain restriction allowed us to obtain a characterization of a maximal domain of choice profiles for which a stable and strategy-proof mechanism exists, which - arguably - is simpler than the corresponding characterization for the general model with non-monotone preferences~\cite[]{Hatfield2018}. This indicates that the domain restriction to monotone preferences removes some complexity from the matching model with contracts and, hence, simplifies its analysis while practical relevance is maintained since the monotonicity assumption is satisfied in many real-world matching markets. We hope that these features will make the model with monotone preferences also a useful tool for future research on matching markets with contracts.

Our results also raise interesting questions for the general matching model with contracts: For our model and choice domain, strategy-proofness of the deferred acceptance mechanism implies its group-strategy-proofness.  This is established through the construction of an equivalent choice profile under which contracts are substitutes. Substitutability implies a lattice structure for the set of stable allocations under the equivalent choice profile, and with the lattice structure it is easy to show that strategy-proofness implies group-strategy-proofness (see e.g.~the discussion in~\citealp{Barbera2016}). For the general matching model with contracts similar constructions, such as the substitutable completion procedure of~\cite{Hatfield2015}, are known. For the same reason as before - the set of stable allocations forms a lattice under the substitutable completion of the choice profile - strategy-proofness of the deferred acceptance mechanism implies its group-strategy-proofness. However, in contrast to our model, for the general model with non-monotone preferences, a substitutable completion does not exist for each profile in the maximal domain of~\cite{Hatfield2018}. It, thus, remains an interesting open question under which conditions the strategy-proofness of a stable matching mechanism for the general model of matching with contracts implies its  group-strategy-proofness.
	\bibliographystyle{RM}
	\bibliography{matching}

\begin{thebibliography}{31}
\newcommand{\enquote}[1]{``#1''}
\providecommand{\natexlab}[1]{#1}
\expandafter\ifx\csname urlstyle\endcsname\relax
  \providecommand{\doi}[1]{doi:\discretionary{}{}{}#1}\else
  \providecommand{\doi}{doi:\discretionary{}{}{}\begingroup
  \urlstyle{rm}\Url}\fi

\bibitem[{Abdulkadiroglu and S{\"o}nmez(2003)}]{AbdulkadirogluSonmez2003}
Abdulkadiroglu, A. and S{\"o}nmez, T. (2003): \enquote{School Choice: A
  Mechanism Design Approach.}
\newblock \emph{American Economic Review}, 93(3): 729--747.

\bibitem[{Abizada and Dur(2017)}]{AbizadaDur2017}
Abizada, A. and Dur, U. (2017): \enquote{College Admissions with
  Complementarities.}

\bibitem[{Alva(2018)}]{alva2018}
Alva, S. (2018): \enquote{WARP and combinatorial choice.}
\newblock \emph{Journal of Economic Theory}, 173(1): 320--333.

\bibitem[{Ayg{\"u}n and S{\"o}nmez(2013)}]{AygunSonmez2012a}
Ayg{\"u}n, O. and S{\"o}nmez, T. (2013): \enquote{Matching with Contracts:
  Comment.}
\newblock \emph{American Economic Review}, 103(5): 2050--2051.

\bibitem[{Barber{\`a} et~al.(2016)Barber{\`a}, Berga, and Moreno}]{Barbera2016}
Barber{\`a}, S., Berga, D., and Moreno, B. (2016): \enquote{Group
  strategy-proofness in private good economies.}
\newblock \emph{American Economic Review}, 106(4): 1073--99.

\bibitem[{Blair(1988)}]{Blair1988}
Blair, C. (1988): \enquote{The Lattice Structure of the Set of Stable Matchings
  with Multiple Partners.}
\newblock \emph{Mathematics of Operations Research}, 18(4): 619--628.

\bibitem[{Chen et~al.(2016)Chen, Egesdal, Pycia, and Yenmez}]{Chenetal2014}
Chen, P., Egesdal, M., Pycia, M., and Yenmez, M.~B. (2016):
  \enquote{Manipulability of Stable Mechanisms.}
\newblock \emph{American Economic Journal: Microeconomics}, 8(2): 202--14.

\bibitem[{Crawford(2008)}]{Crawford2008}
Crawford, V.~P. (2008): \enquote{The flexible-salary match: a proposal to
  increase the salary flexibility of the national resident matching program.}
\newblock \emph{Journal of Economic Behavior \& Organization}, 66(2): 149--160.

\bibitem[{Echenique(2012)}]{Echenique2012}
Echenique, F. (2012): \enquote{Contracts versus Salaries in Matching.}
\newblock \emph{American Economic Review}, 102(1): 594--601.

\bibitem[{Fleiner(2003)}]{Fleiner2003}
Fleiner, T. (2003): \enquote{A Fixed-Point Approach to Stable Matchings and
  some Applications.}
\newblock \emph{Mathematics of Operations Research}, 28(1): 103--126.

\bibitem[{Hassidim et~al.(2017)Hassidim, Romm, and Shorrer}]{Hassidim2017}
Hassidim, A., Romm, A., and Shorrer, R.~I. (2017): \enquote{Need vs. Merit: The
  Large Core of College Admissions Markets.}

\bibitem[{Hassidim et~al.(2020)Hassidim, Romm, and Shorrer}]{hassidim2016}
Hassidim, A., Romm, A., and Shorrer, R.~I. (2020): \enquote{The Limits of
  Incentives in Economic Matching Procedures.}
\newblock \emph{Management Science}, forthcoming.

\bibitem[{Hatfield and Kojima(2008)}]{HatfieldKojima2008}
Hatfield, J.~W. and Kojima, F. (2008): \enquote{Matching with contracts:
  Comment.}
\newblock \emph{American Economic Review}, 98(3): 1189--94.

\bibitem[{Hatfield and Kojima(2009)}]{HatfieldKojima2009}
Hatfield, J.~W. and Kojima, F. (2009): \enquote{Group Incentive Compatibility
  for Matching with Contracts.}
\newblock \emph{Games and Economic Behavior}, 67(2): 745--749.

\bibitem[{Hatfield and Kojima(2010)}]{HatfieldKojima2010}
Hatfield, J.~W. and Kojima, F. (2010): \enquote{Substitutes and Stability for
  Matching with Contracts.}
\newblock \emph{Journal of Economic Theory}, 145(5): 1704--1723.

\bibitem[{Hatfield and Kominers(2015)}]{Hatfield2015}
Hatfield, J.~W. and Kominers, S.~D. (2015): \enquote{Hidden Substitutes.}
\newblock In \emph{Proceedings of the Sixteenth ACM Conference on Economics and
  Computation}, EC '15, pages 37--37. ACM, New York, NY, USA.

\bibitem[{Hatfield et~al.(2020)Hatfield, Kominers, and Westkamp}]{Hatfield2018}
Hatfield, J.~W., Kominers, S.~D., and Westkamp, A. (2020): \enquote{Stability,
  Strategy-Proofness, and Cumulative Offer Mechanisms.}
\newblock \emph{Review of Economic Studies}, forthcoming.

\bibitem[{Hatfield and Milgrom(2005)}]{HatfieldMilgrom2005}
Hatfield, J.~W. and Milgrom, P.~R. (2005): \enquote{Matching with Contracts.}
\newblock \emph{American Economic Review}, 95(4): 913--935.

\bibitem[{Hirata and Kasuya(2014)}]{HirataKasuya2014}
Hirata, D. and Kasuya, Y. (2014): \enquote{Cumulative Offer Process is
  Order-Independent.}
\newblock \emph{Economics Letters}, 124(1): 37--40.

\bibitem[{Hirata and Kasuya(2017)}]{HirataKasuya2015}
Hirata, D. and Kasuya, Y. (2017): \enquote{On Stable and Strategy-proof Rules
  in Matching Markets with Contracts.}
\newblock \emph{Journal of Economic Theory}, 168: 27 -- 43.

\bibitem[{Jagadeesan(2019)}]{Jagadeesan2016}
Jagadeesan, R. (2019): \enquote{Cadet-Branch Matching in a Kelso-Crawford
  Economy.}
\newblock \emph{American Economic Journal: Microeconomics}, 11(3): 191--224.

\bibitem[{Kelso and Crawford(1982)}]{KelsoCrawford1982}
Kelso, A. and Crawford, V.~P. (1982): \enquote{Job Matching, Coalition
  Formation, and Gross Substitutes.}
\newblock \emph{Econometrica}, 50(6): 1483--1504.

\bibitem[{Kominers(2012)}]{Kominers2011}
Kominers, S.~D. (2012): \enquote{On the Correspondence of Contracts to Salaries
  in (Many-to-Many) Matching.}
\newblock \emph{Games and Economic Behavior}, 75(2): 984--989.

\bibitem[{Kominers and S{\"o}nmez(2016)}]{KominersSonmez2015}
Kominers, S.~D. and S{\"o}nmez, T. (2016): \enquote{Matching with Slot-Specific
  Priorities: Theory.}
\newblock \emph{Theoretical Economics}, 11: 683--710.

\bibitem[{Roth(1984{\natexlab{a}})}]{Roth1984a}
Roth, A.~E. (1984{\natexlab{a}}): \enquote{The Evolution of the Labor Market
  for Medical Interns and Residents: A Case Study in Game Theory.}
\newblock \emph{Journal of Political Economy}, 92(6): 991--1016.

\bibitem[{Roth(1984{\natexlab{b}})}]{Roth1984b}
Roth, A.~E. (1984{\natexlab{b}}): \enquote{Stability and Polarization of
  Interests in Job Matching.}
\newblock \emph{Econometrica}, 52(1): 47--57.

\bibitem[{Roth(1991)}]{Roth.1991}
Roth, A.~E. (1991): \enquote{A Natural Experiment in the Organization of
  Entry-Level Labor Markets: Regional Markets for New Physicians and Surgeons
  in the United Kingdom.}
\newblock \emph{American Economic Review}, 81: 415--440.

\bibitem[{Schlegel(2015)}]{Schlegel2015}
Schlegel, J.~C. (2015): \enquote{Contracts versus Salaries in Matching: A
  General Result.}
\newblock \emph{Journal of Economic Theory}, 159(A): 552--573.

\bibitem[{Schlegel(2016)}]{Schlegel2016}
Schlegel, J.~C. (2016): \enquote{Virtual Demand and Stable Mechanisms.}
\newblock In \emph{Proceedings of the 2016 ACM Conference on Economics and
  Computation}, pages 63--63. ACM.

\bibitem[{S{\"o}nmez(2013)}]{Sonmez2013}
S{\"o}nmez, T. (2013): \enquote{Bidding for Army Career Specialities: Improving
  the ROTC Branching Mechanism.}
\newblock \emph{Journal of Political Economy}, 121(1): 186--219.

\bibitem[{S{\"o}nmez and Switzer(2013)}]{SonmezSwitzer2013}
S{\"o}nmez, T. and Switzer, T.~B. (2013): \enquote{Matching With
  (Branch-of-Choice) Contracts at the United States Military Academy.}
\newblock \emph{Econometrica}, 81(2): 451--488.

\end{thebibliography}
	\appendix
	
	\section{Proof of Theorem~\ref{smoothing}}
	The proof relies on the following two lemmata that we prove first.
	\begin{lemma}\label{order1}
		If contracts are observable substitutes  under $Ch_c$, then for each $Y\subseteq X_c$ and for any two maximal sequences $x^1,\ldots,x^{\tau}$ and $y^1,\ldots,y^{\tau'}$
		in $Y$ that are observable under $Ch_c$ and generated from monotone
		preferences we have
		$\{x^1,\ldots,x^{\tau}\}=\{y^1,\ldots,y^{\tau'}\}$.
	\end{lemma}
	\begin{proof}
		The proof strategy is due to~\cite{HirataKasuya2014}. 
		We use induction on the size of the set $Y$.
		\newline
		\noindent {\bf Induction Basis:} If $|Y|=0$, then $Y=\emptyset$ and trivially the empty sequence is the only observable sequence under $Ch_c$ generated from monotone preferences in $Y$.     \newline
		\noindent {\bf Induction Assumption:} For each $Y\subseteq  X_s$ with
		$|Y|\leq n$  for any two maximal sequences $x^1,\ldots,x^{\tau}$ and $y^1,\ldots,y^{\tau'}$
		in $Y$ that are observable under $Ch_c$ and generated from monotone
		preferences we have
		$\{x^1,\ldots,x^{\tau}\}=\{y^1,\ldots,y^{\tau'}\}$.
		\newline
		\noindent {\bf Induction Step:}
		Let $|Y|=n+1$.	Consider two maximal sequences $x^1,\ldots,x^{\tau}$ and $y^1,\ldots,y^{\tau'}$
		in $Y$ that are observable under $Ch_c$ and generated from monotone
		preferences. Suppose
		$\{x^1,\ldots,x^{\tau}\}\neq\{y^1,\ldots,y^{\tau'}\}$.  Then
		w.l.o.g.~there is a $1\leq t\leq \tau'$ with
		$y^t\notin\{x^1,\ldots,x^{\tau}\}$. Choose the smallest such $t$ and
		consider the set $Y':=Y\setminus\{y^t\}$. 
		The sequence $y^1,\ldots,y^{t-1}$ is contained in $Y'$, generated from monotone preferences and is observable
		under $Ch_c$. Extend $y^1,\ldots y^{t-1}$ to a maximal sequence
		$y^1,\ldots,y^{t-1},\tilde{y}^t,\ldots,\tilde{y}^{\tilde{\tau}}$ in
		$Y'$ that is observable under $Ch_c$ and generated from monotone
		preferences.  Since
		$y^t\notin\{x^1,\ldots,x^{\tau}\}$, sequence $x^1,\ldots,x^{\tau}$ is a maximal sequence in $Y'$ that is observable under $Ch_c$ and generated from monotone preferences. By the induction assumption, we have
		$\{y^1\,\ldots,y^{t-1},\tilde{y}^{t},\ldots,\tilde{y}^{\tilde{\tau}}\}=\{x^1,\ldots,x^{\tau}\}$.
		Moreover, $y^t_{S}\in Ch_c\{x^1,\ldots,x^{\tau}\}_{S}$, as
		otherwise $x^1,\ldots,x^{\tau},y^t$ would be observable under $Ch_c$ and generated from
		monotone preferences, contradicting the maximality of $x^1,\ldots,x^{\tau}$ in $Y$. Since
		$y^1,\ldots,y^t$ is observable and generated from monotone preferences, we have $y\in R_c\{y^1,\ldots,y^{t-1}\}$ for each $y\in X_c$ with $y_S=y^t_S$ and $y_T\vartriangleright y^t_T$. By observable substitutability, we	have  $y\in
			R_c\{y^1\,\ldots,y^{t-1},\tilde{y}^{t},\ldots,\tilde{y}^{\tilde{\tau}}\}$ for each $y\in X_c$ with $y_S=y^t_S$ and $y_T\vartriangleright y^t_T$. As $y^t\notin \{y^1,\ldots,y^{t-1},\tilde{y}^{t},\ldots,\tilde{y}^{\tilde{\tau}}\}$ this implies $y^t_{S}\notin
		Ch_c\{y^1\,\ldots,y^{t-1},\tilde{y}^{t},\ldots,\tilde{y}^{\tilde{\tau}}\}_{S}$. Since $y^t_{S}\in Ch_c\{x^1,\ldots,x^{\tau}\}_{S}$, as previously observed, this contradicts
		$$\{y^1\,\ldots,y^{t-1},\tilde{y}^{t},\ldots,\tilde{y}^{\tilde{\tau}}\}=\{x^1,\ldots,x^{\tau}\}.$$
		
	\end{proof}
	\begin{lemma}\label{order2}
		If contracts are observable substitutes under $Ch_c$, then
		\begin{enumerate}
			\item for $Y,Z\in\mathcal{O}(Ch_c)$ with $Z\subseteq Y$ we have $$R_c(Z)\subseteq R_c(Y).$$ If moreover, the observable law of aggregate demand holds for $Ch_c$, then $$|Ch_c(Z)|\leq|Ch_c(Y)|,$$
			\item for $Ch_{-c}\in\bigtimes_{c'\in C\setminus\{c\}}\mathcal{C}_c$ and $\succeq\in\bigtimes_{s\in S}\mathcal{R}_s$ the sets of contracts $Y_1\subseteq Y_2\subseteq \ldots\subseteq Y_{\tau}\subseteq X_c$ proposed to $c$ during the cumulative offer algorithm in $(Ch,\succeq)$ are observable under $Ch_c$.

		\end{enumerate}
	\end{lemma}
	\begin{proof}

		For the first part let $Y,Z\in\mathcal{O}(Ch_c)$ with $Z\subseteq Y$. Since $Z$ is observable, we can find a sequence $x^1,\ldots,x^{|Z|}$ generated from monotone preferences with $Z=\{x^1,\ldots,x^{|Z|}\}$ that is observable under $Ch_c$. Maximally extend the sequence in $Y$ to obtain a sequence $x^1,\ldots,x^{|Z|},x^{|Z|+1},\ldots,x^{\tau}$ generated from monotone preferences that is observable under $Ch_c$.  By Lemma~\ref{order1} and observability of $Y$ we have $Y=\{x^1,\ldots,x^{\tau}\}$ and $\tau=|Y|$. The result follows by applying observable substitutability (resp. the observable law of aggregate demand) to the sequence $x^1,\ldots,x^{|Y|}.$
		
		For the second part, note that for $t=1$ the result holds trivially: by the definition of the cumulative offer algorithm each contract in $Y_1$ is the unique most preferred contract for the involved agent among all contracts. Thus, ordering the contracts in $Y_1$ arbitrarily yields an observable sequence under $Ch_c$ generated from monotone preferences.
		Next let $1< t\leq \tau$ and suppose that $Y_1,\ldots,Y_{t-1}\in\mathcal{O}(Ch_c)$. Take an observable sequence $x^{1},\ldots,x^{|Y_{t-1}|}$ under $Ch_c$ generated from monotone preferences such that $Y_{t-1}=\{x^1,\ldots,x^{|Y_{t-1}|}\}$ and enumerate contracts in $Y_{t}\setminus Y_{t-1}=\{x^{|Y_{t-1}|},\ldots x^{|Y_t|}\}$ arbitrarily. By the definition of the CO algorithm, for each $x\in Y_t\setminus Y_{t-1}$ each $x'\in X_c$ with $x'_S=x_s$ and $x'\vartriangleright x$ has been rejected in some previous round of the algorithm, i.e. there is a $t'<t$ with $x'\in R_c(Y_{t'})$. By observable sustitutability and the second part applied to $Z=Y_{t'}$ and $Y=Y_{t-1}$, we have $x'\in R_c(Y_{t-1})=R_c\{x^1,\ldots,x^{|Y_{t-1}|}\}$. Thus, for each $x\in Y_t\setminus Y_{t-1}$ we have $x_S\notin Ch_c(Y_{t-1})_S$. Observe furthermore that $Y_t\setminus Y_{t-1}$ contains at most one contract per student. Thus, repeated application of observable substitutability shows that the sequence $x^1,\ldots x^{|Y_t|}$ is observable.
	\end{proof}

	With the two lemmata we can prove the theorem.
	
	\begin{proof}
		For each $Y\subseteq X_c$, let $Y^{\vee}\in\mathcal{O}(Ch_c)$ be defined by $Y^{\vee}:=\{x^1,\ldots,x^{\tau}\}$ where $x^1,\ldots,x^{\tau}$ is a maximal observable subsequence under $Ch_c$ of $\mathcal{U}(Y)$ generated by monotone preferences. By Lemma~\ref{order1}, $Y^{\vee}$ is well-defined and $\tilde{Y}\subseteq Y^{\vee}$ for each $\tilde{Y}\in\mathcal{O}(Ch_c)$ with $\tilde{Y}\subseteq \mathcal{U}(Y)$. Thus, by the first part of Lemma~\ref{order2},
		$$\bigcup_{\tilde{Y}\subseteq \mathcal{U}(Y),\tilde{Y}\in\mathcal{O}(Ch_c)}R_c(\tilde{Y})=R_c(Y^{\vee}),$$
		and therefore the virtual choice function is given by $$Ch_c'(Y)=Y^{\min}\setminus R_c(Y^{\vee}).$$
		First note that $Ch_c'$ satisfies our assumptions on choice
		functions: By definition $Ch_c'(Y)\subseteq Y^{\min}\subseteq Y$,
		and since $Y^{\min}$ contains at most one contract per student, also
		$Ch'_c(Y)$ contains at most one contract per student. The IRC
		condition for $Ch_c'$ will follow from substitutability and the law
		of aggregate demand (see~\citealp{AygunSonmez2012a}) for $Ch'_c$ which we
		will establish next.
		
		Let $Z\subseteq Y\subseteq  X_c$. First we show substitutability, i.e.~$R_c'(Z)\subseteq R_c'(Y)$. Let $x\in R_c'(Z)$. If $x\in Y\setminus Y^{\min}$, then $x\in R'_c(Y)$, as $Ch_c'(Y)\subseteq Y^{\min}$. If $x\in Y^{\min}$, then $x\in Z^{\min}$ and therefore $x\in R_c(\tilde{Z})$ for a $\tilde{Z}\in \mathcal{O}(Ch_c)$ with $\tilde{Z}\subseteq \mathcal{U}(Z)\subseteq\mathcal{U}(Y)$. Thus, also $x\in R_c'(Y)$.

		By definition of $Ch_c'$ and $Y^{\vee}$ (resp.~$Z^{\vee}$) we have $Ch_c'(Y)_S=Ch_c(Y^{\vee})_S$ and $Ch_c'(Z)_S=Ch_c(Z^{\vee})_S$. Lemma~\ref{order1} implies that $Z^{\vee}\subseteq Y^{\vee}$. Thus, by the first part of Lemma~\ref{order2} we 
		obtain the law of aggregate demand for $Ch'_c$, as
		$$|Ch'_c(Z)|=|Ch'_c(Z)_S|=|Ch_c(Z^{\vee})|\leq|Ch_c(Y^{\vee})|=|Ch'_c(Y)_S|=|Ch'_c(Y)|.$$ 
			
		To show that $Ch_c'$ is equivalent to $Ch_c$, let $Ch_{-c}\in\bigtimes_{c'\neq c}\mathcal{C}_c$ and $\succeq\in\bigtimes_{s\in S}\mathcal{R}_c$. By the second part of Lemma~\ref{order2}, the sets of proposed contracts $Y_1\subseteq Y_2\subseteq \ldots\subseteq Y_{\tau}$ to $c$ during the CO-algorithm in $(Ch,\succeq)$ are observable.
		Now note that for each $Y\in\mathcal{O}(Ch_c)$ we have $Y=Y^{\vee}$ and, moreover, by observable substitutability $Ch_c(Y)\subseteq Y^{\min}$. Thus, $Ch'_c(Y)=Y^{\min}\setminus R_c(Y)= Y\setminus R_c(Y)=Ch_c(Y)$. As $Ch_c'$ and $Ch_c$ agree on observable sets, we have $Ch_c'(Y_t)=Ch_c(Y_t)$ for $t=1,\ldots, \tau$. Hence $Ch_c'$ and $Ch_c$ are equivalent.
	\end{proof}
	\section{Proof of Theorem~\ref{andersrum}}
	\begin{proof}
		
		Let $c\in C$ and consider a sequence $x^1,\ldots,x^{\tau}$
		that is observable under $Ch_c$ and generated from monotone
		preferences. We show that
		$Ch_c'\{x^1,\ldots,x^{\tau}\}=Ch_c\{x^1,\ldots,x^{\tau}\}$. Since contracts are substitutes under
		$Ch_c'$ this will imply that contracts are observable substitutes under $Ch_c$. Moreover, if the law of aggregate demand holds under $Ch'_c$, then this will imply the observable law of aggregate demand for $Ch_c$. In the following,
		we denote by $\succeq^0\in\bigtimes_{s\in S}\mathcal{R}_s$ a profile
		such that no contract is acceptable, and for $1\leq t\leq {\tau}$ we
		denote by $\succeq^t\in\bigtimes_{s\in S}\mathcal{R}_s$ a profile such that $x\succ^t_{x_{S}}\emptyset$ for
		$x\in\{x^1,\ldots,x^t\}$ and $\emptyset\succ^t_{x_{S}}x$ for
		$x\notin\{x^1,\ldots,x^t\}$. First note that
		$$\emptyset=Ch_c'(\emptyset)=\mathcal{CO}(Ch',\succeq^0)=\mathcal{CO}(Ch',\succeq^0)=Ch_s(\emptyset).$$
		Next we show that if for $0\leq t<t'\leq \tau$ we have
		$$Ch'_c\{x^1,\ldots,x^{t}\}=\mathcal{CO}(Ch',\succeq^t)=\mathcal{CO}(Ch,\succeq^t)=Ch_c\{x^1,\ldots,x^{t}\}$$
		then
		$$Ch'_c\{x^1,\ldots,x^{t'}\}=\mathcal{CO}(Ch',\succeq^{t'})=\mathcal{CO}(Ch,\succeq^{t'})=Ch_c\{x^1,\ldots,x^{t'}\}.$$
		Since $Ch_{c}\{x^1,\ldots,x^{t}\}=Ch'_c\{x^1,\ldots,x^{t}\}$ for
		$0\leq t<t'$, observability of  $x^1,\ldots,x^{t'}$ under $Ch_c$
		implies observability of  $x^1,\ldots,x^{t'}$ under $Ch'_c$.
		
		Let
		$Y$ be the set of proposed contract during the CO-algorithm in $(Ch',\succeq^{t'})$. By (observable) substitutability of $Ch_c'$, the second part of Lemma~\ref{order2} applied to $(Ch',\succeq^{t'})$  implies that $Y$ is observable under $Ch'_c$. Thus, there is a sequence $y^1,\ldots,y^{t''}$ that is observable under $Ch_c$ and generated from monotone preferences such that $Y=\{y^1,\ldots,y^{t''}\}\subseteq\{x^1,\ldots,x^{t'}\}$. Note furthermore that $y^1,\ldots,y^{t''}$ is a maximal such sequence in the set $\{x^1,\ldots,x^{t'}\}$, since otherwise not all students who are unmatched in the outcome of the CO algorithm in $(Ch',\succeq^{t'})$ have proposed under all acceptable contracts.
		Since $x^1,\ldots,x^{t'}$ is observable under $Ch_c'$ this implies $Y=\{y^1,\ldots,y^{t''}\}=\{x^1,\ldots,x^{t'}\}$.
		By equivalence of $Ch$ and $Ch'$, we have 	$$Ch'_c\{x^1,\ldots,x^{t'}\}=\mathcal{CO}(Ch',\succeq^{t'})=\mathcal{CO}(Ch,\succeq^{t'})=Ch_c\{x^1,\ldots,x^{t'}\}.$$
		
	\end{proof}
	
\end{document}